\documentclass[a4paper,12pt]{article}

\usepackage{amsthm}
\usepackage{amsmath}
\usepackage{epsfig}
\usepackage{amssymb}
\usepackage{latexsym}
\usepackage[T1]{fontenc}
\usepackage[utf8]{inputenc}
\usepackage{color}
\usepackage{url}

\newtheorem{definition}{Definition}
\newtheorem{theorem}[definition]{Theorem}

\newtheorem{lemma}[definition]{Lemma}
\newtheorem{conjecture}{Conjecture}

\newtheorem{question}[conjecture]{Question}
\let\epsilon=\varepsilon

\let\rho = \varrho

\begin{document}

\title{Weak oddness as an approximation of oddness and resistance in cubic graphs}

\author{
Robert Lukoťka and Ján Mazák
\\[3mm]
\\{\tt \{lukotka, mazak\}@dcs.fmph.uniba.sk}
\\[5mm]
Comenius University, Mlynská dolina, 842 48 Bratislava\\
}

\maketitle

\begin{abstract}
We introduce \emph{weak oddness} $\omega_{\textrm w}$, a new measure of uncolourability of cubic graphs, defined as the least number of odd components in an even factor.
For every bridgeless cubic graph $G$, $\rho(G)\le\omega_{\textrm w}(G)\le\omega(G)$, 
where $\rho(G)$ denotes the resistance of $G$ and $\omega(G)$ denotes the oddness of $G$, so this new measure is an approximation of both oddness and resistance.
We demonstrate that there are graphs $G$ satisfying $\rho(G) < \omega_{\textrm w}(G) < \omega(G)$, and that the difference between any two of those three measures can be arbitrarily large.
The construction implies that if we replace a vertex of a cubic graph with a triangle, then its oddness can decrease by an arbitrary large amount.
\end{abstract}

\section{Oddness and resistance}

Cubic graphs naturally fall into two classes depending on
whether they do or do not admit a $3$-edge-colouring.
Besides the trivial family of graphs with bridges, which are trivially uncolourable, there are many examples of $2$-edge-connected cubic graphs that do not admit a $3$-edge-colouring.
Such graphs are called \emph{snarks}; sometimes they are required to
satisfy additional conditions, such as cyclic
$4$-edge-connectivity and girth at least five, to avoid
triviality.

In their many attempts to understand snarks better, researchers have come up with various measures that refine the notion of ``being close to 3-edge-colourable''.
In Section~\ref{sec2}, we introduce a new such measure closely related to oddness and resistance. We follow \cite{my} in our presentation of these two concepts.

Every bridgeless cubic graph has a
$1$-factor \cite{petersen} and consequently also a $2$-factor.
It is easy to see that a cubic graph is $3$-edge-colourable if
and only if it has a $2$-factor that only consists of even
circuits. In other words, snarks are those cubic graphs which
have an odd circuit in every $2$-factor. The minimum number of
odd circuits in a $2$-factor of a bridgeless cubic graph $G$ is
its \emph{oddness}, and is denoted by $\omega(G)$. Since every
cubic graph has even number of vertices, its oddness must also
be even.

The relevance of oddness stems from the importance of snarks.
The crux of many important problems and conjectures, like the Tutte's $5$-flow conjecture or the cycle
double cover conjecture, consists in dealing with snarks.
While most of these
problems are exceedingly difficult for snarks in general, they are often tractable for those that are close to being $3$-edge-colourable.
For example, the $5$-flow conjecture has been
verified for snarks with oddness at most $2$ (Jaeger
\cite{jaeger2}) and for cyclically $6$-edge-connected snarks with oddness at most $4$ \cite{sm}, and the cycle double cover conjecture has been
verified for snarks with oddness at most $4$ (Huck and Kochol
\cite{hk}, H\"aggkvist and McGuinness \cite{hg}). Snarks with
large oddness thus remain potential counterexamples to these
conjectures and therefore deserve further study.

Another natural measure of uncolourability of a cubic graph is
based on minimising the use of the fourth colour in a
$4$-edge-colouring of a cubic graph, which can alternatively be viewed as minimising the number of edges that have to be deleted in order to get a
$3$-edge-colourable graph. Surprisingly, the required
number of edges to be deleted is the same as the number of
vertices that have to be deleted in order to get a
$3$-edge-colourable graph (see \cite[Theorem~2.7]{steffen1}).
This quantity is called the \emph{resistance} of $G$, and will
be denoted by~$\rho(G)$. Observe that $\rho(G) \le \omega(G)$
for every bridgeless cubic graph $G$ since deleting one edge
from each odd circuit in a $2$-factor leaves a colourable
graph. The difference between $\rho(G)$ and $\omega(G)$ can be arbitrarily large in
general \cite{steffen2}.

\section{Weak oddness}
\label{sec2}

A \emph{factor} of a graph $G$ is a subgraph containing all vertices of $G$. The most commonly encountered factors are those defined by vertex degrees: for instance, a $1$-factor is a spanning subgraph whose all vertices have degree $1$. There is, however, an almost endless list of other variants of factors; we refer the reader to the excellent surveys \cite{kano, plummer} for an overview. One important kind of factors are \emph{even factors}, that is, factors with all vertices of even degree \cite{kano, tsp, plummer, xiong}. Even factors play a profound role in the areas of flows and cycle covers where snarks are especially relevant. It is, therefore, very natural to introduce a measure of uncolourability based on even factors: the \emph{weak oddness} of a cubic graph $G$, denoted by $\omega_{\textrm w}(G)$, is the least number of odd components in an even factor of $G$. (We do not need to avoid graphs with bridges in the definition since every graph has an even factor containing just isolated vertices.)

In a cubic graph, an even factor is comprised of circuits and isolated vertices, thus it can be viewed as a relaxation of $2$-factor where ``degenerated'' circuits of length $1$ are allowed. 
Since every $2$-factor is an even factor, $\omega_{\textrm w}\le \omega$. On the other hand, we can remove a vertex from each odd component of an even factor of a graph $G$, and thus obtain a $3$-edge-colourable graph, so $\rho\le \omega_{\textrm w}$. 
The new invariant $\omega_{\textrm w}$ thus approximates resistance from above more tightly than oddness, and approximates oddness from below more tightly than resistance. Of course, the enhanced approximation potential materializes only if weak oddness differs from both oddness and resistance. We demonstrate that this is indeed the case in Section~\ref{construction}.

Besides that, there are several reasons why we think weak oddness is useful.
First, weak oddness does not change when we contract a triangle. On the other hand, as we will show in Section~\ref{conc}, contracting a triangle can increase oddness arbitrarily.
We hope that this property might improve the chance of developing inductive methods based on oddness.

Second, many results for graphs with small oddness can be easily transformed to results for graphs with small weak oddness. As $\omega$ might be larger than $\omega_{\textrm w}$, the modified statements comprise a larger family of graphs. As a potential example we give the following.
\begin{theorem}\label{cdcthm}
Graphs with weak oddness at most $4$ have a cycle-double cover.
\end{theorem}
\begin{proof}
Let $G$ be a graph with weak oddness $4$. Let us replace each vertex of $G$ with a triangle,
obtaining a graph $G'$. Every even factor of $G$ can be naturally extended to a $2$-factor of $G'$, thus the oddness of $F'$ is at most $4$.
As graphs with oddness at most $4$ have a cycle double cover \cite{hg, hk}, the graph $G'$ has a cycle double cover. This cover
reduced to the edges of $G$ is a cycle double cover of $G$.
\end{proof}
It is known that if $\rho(G) \le 2$ (and thus $\omega_{\textrm w}(G)\le\omega(G)\le 2$), then
$\rho(G)=\omega(G)$ \cite[Lemma 2.5]{steffen1} and thus $\omega_{\textrm w}(G)=\omega(G)$.
For graphs with resistance more than $2$, resistance and oddness may be distinct 
\cite{steffen2}. 
In Section~\ref{construction} we find a graph with weak oddness $14$ and oddness $16$ which illustrates that weak oddness and oddness can differ.
We do not know, however, whether weak odness and oddness can be different for graphs with 
weak oddness smaller than $14$. In particular, the following is an open problem.
\begin{question}
Does there exist a cubic graph with weak oddness $4$ and oddness at least $6$?
\end{question}
\noindent If the answer to this question is affirmative, then Theorem~\ref{cdcthm} is more general than the original theorem requiring oddness at most $4$.

\section{Graphs with $\rho < \omega_{\textrm w} < \omega$}
\label{construction}

Our construction utilizes smaller blocks to build larger graphs. 
A \emph{$2$-pole} is a triple $(G, s, t)$, where $G$ is a graph and $s$, $t$ are two different vertices of $G$ which both have degree one;
we will call them \emph{terminals} and the edges incident to them \emph{terminal edges}.
Each terminal of a $2$-pole serves as a place of connection with
another terminal. Two terminal edges of two disjoint $2$-poles can be naturally joined to form a
new nonterminal edge by identifying the terminal
vertices incident to them and suppressing the resulting $2$-valent vertex. 

A standard way to create terminal vertices is by \emph{splitting off} a vertex $v$ from a
graph~$G$; by this we mean removing $v$ from $G$ and
attaching a terminal vertex to each dangling edge originally
incident with $v$. 

The first step in our construction is to create the $2$-pole $H$ depicted in Fig.~\ref{diamond}. Let $P$ be the $2$-pole obtained from the Petersen graph by inserting a new vertex into one of its edges and then splitting the new vertex off. (The Petersen graph is edge-transitive, so the result of our operation is uniquely determined.) We take two copies $(P_1, s_1, t_1)$, $(P_2, s_2, t_2)$ of $P$, identify $s_1$ with $s_2$, $t_1$ with $t_2$, and attach a new terminal edge to both of $s_1$ and $t_1$.

\begin{figure}
\centering\includegraphics{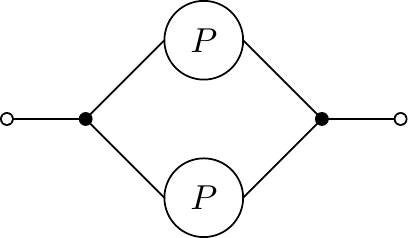}
\caption{The $2$-pole $H$ (its terminal vertices are marked by empty circles).}
\label{diamond}
\end{figure}

The next step is to create the $2$-pole $H_2$ by taking two copies of $P$ and joining a terminal edge of the first copy to a terminal edge of the second copy.

Finally, we take the complete graph on four vertices $u_0$, $u_1$, $u_2$, $u_3$ and remove all edges incident with $u_0$.
Then, for each $i\in\{1,2,3\}$, we take a new copy of $H_2$ and identify one of its terminals with $u_0$ and the other with $u_i$.
We denote the resulting cubic graph $G$ (see Fig.~\ref{thegraph}).

\begin{figure}
\centering\includegraphics{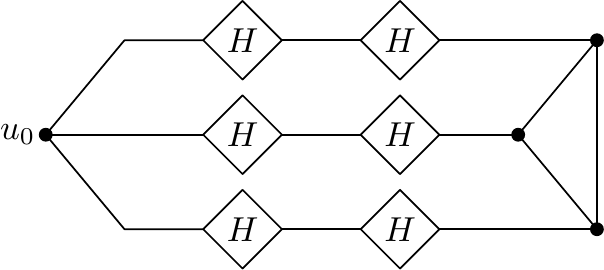}
\caption{The construction of the graph $G$.}
\label{thegraph}
\end{figure}

\begin{lemma}
If an even factor $F$ of $G$ contains a cycle $C$ passing through a non-terminal vertex of a copy $H'$ of $H$, but not lying in $H'$, then $H'$ contains at least three odd components of $F$ different from $C$.
\label{passing}
\end{lemma}

\begin{proof}
If $C$ passes through $H$, it passes through non-terminal vertices of exactly one copy of $P$ contained in $H$.
In that copy, there is at least one odd component of $F$ apart from $C$.
In the other copy of $P$, there must be at least two odd components, because the Petersen is not $3$-edge-colourable.
\end{proof}

\begin{theorem}
The resistance, weak oddness, and oddness of the graph $G$ are $12$, $14$, and $16$, respectively.
\label{main}
\end{theorem}

\begin{proof}
First, we determine the resistance of $G$. Consider a copy of $P$ included in $H$. If we remove a vertex at distance $3$ from both terminals from it, the rest is $3$-edge-colourable in such a way that the colours of terminal edges of $P$ are different (one possible colouring is presented in Fig.~\ref{colouring}).
Consequently, it is possible to remove two vertices of $H$ (one from each copy of $P$) and properly colour the edges of the remaining graph with $3$ colours so that the terminal edges of $H$ have the same colour. In other words, $H$ behaves just like an edge for our purpose, and so does $H_2$, so we are essentially colouring just $K_4$.
Therefore, by removing $12$ vertices from $G$ (two from each copy of $P$), we obtain a $3$-edge-colourable graph. 
On the other hand, if we do not remove a vertex from $P$, it is not $3$-edge-colourable. This proves that $\rho(G) = 12$.

Assume that $F$ is an even factor of $G$ with minimum number of odd cycles. If $u_0$ is an isolated vertex in $F$, then there are at least $12$ odd circuits in the copies of $P$, and one more triangle in the rest of the graph. If $u_0$ is not isolated in $F$, then some four copies of $H$ are passed through by a circuit of $F$, so, by Lemma~\ref{passing}, they contain at least $4\cdot 3$ odd circuits. Each of the remaining two copies of $H$ contain at least two odd circuits. In either case, $F$ has at least $14$ odd components, and so $\omega_{\textrm w}(G)\ge 14$.

Next, we describe an even factor $F$ of $G$ which contains $u_0$ as an isolated vertex and has $14$ odd components. No copy of $H$ is passed through by $F$. The Petersen graph has a $2$-factor comprised of two $5$-circuits, and thus $H$ has a $2$-factor composed of two $5$-circuits (one in each copy of $P$) and one $12$-circuit. Altogether, $F$ is comprised of six $12$-circuits, twelve $5$-circuits, one triangle $u_1u_2u_3$, and one isolated vertex. Consequently, $\omega_{\textrm w}\le 14$. 

Finally, we determine the oddness of $G$. In every $2$-factor of $G$, the vertex $u_0$ belongs to a circuit of $F$, and this circuit must thus pass through four copies of $H$.
According to Lemma~\ref{passing}, those copies of $H$ contain at least $12$ odd circuits of $F$, and each of the remaining two copies of $H$ contains at least two odd circuits of $F$. Consequently, there are at least $16$ odd cycles, so $\omega(G)\ge 16$. It is easy to find a $2$-factor of $F$ having $16$ odd cycles (actually, all $2$-factors of $G$ have this property).
\end{proof}

\begin{figure}
\centering\includegraphics{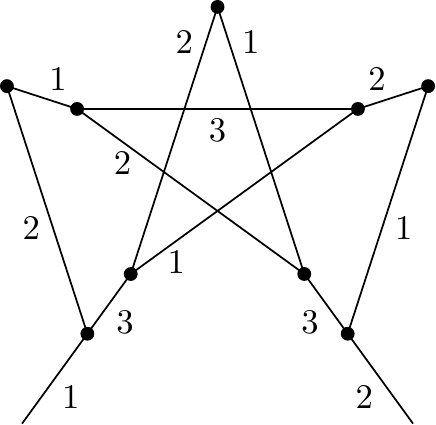}
\caption{A $3$-edge-colouring of $P$ without a vertex.}
\label{colouring}
\end{figure}

As can be easily seen from the proof of Theorem~\ref{main}, the construction can be modified by using more copies of $H$ in $H_2$,
which would lead to an arbitrarily large difference between oddness and weak oddness.
If we also insert one copy of $H_2$ into each of the edges $u_1u_2$, $u_2u_3$ and $u_3u_1$, the difference between weak oddness and resistance grows to $4$:
it is enough to delete one vertex in each copy of $P$ in order to get a $3$-edge-colourable graph, but we have to include all the isolated vertices $u_0$, $u_1$, $u_2$, $u_3$ in an even factor to have only one odd component in each copy of $P$. For an even larger difference, we can take 
a large $3$-edge-colourable cubic graph and insert a copy of $H_2$ into each of its edges.

Besides showing that resistance, weak oddness, and oddness can be all distinct,
we can also make the following observation.
If we expand the vertex $u_0$ of $G$ into a triangle, the resulting graph will have oddness $14$, 
while $G$ has oddness $16$. Thus expanding a vertex into a triangle can decrease oddness.
By using more copies of $H$ to produce $H_2$ we can obtain a graph 
in which the expansion of $u_0$ into a triangle can decrease the oddness by as much as we wish (without changing its parity, of course).

\section{Conclusion}
\label{conc}

While we have demonstrated that weak oddness might differ from oddness, we have very little insight into a characterisation of graphs for which $\omega=\omega_{\textrm w}$.
Apparently, a small snark $G$ has $\omega=\omega_{\textrm w}=2$ (this is true for snarks up to at least $26$ vertices \cite{my} and cyclically $4$-connected snarks up to at least $36$ vertices \cite{brinkmann}). On the one hand, a snark with edge cuts of size $2$ can have vastly different values of $\omega$ and $\omega_{\textrm w}$. On the other hand, if the Jaeger's conjecture \cite{jaeger7cc} is true, then there are no cyclically $7$-connected snarks, so for cubic graphs with sufficiently large cyclic connectivity it might be well possible that $\omega=\omega_{\textrm w}=0$. We propose the following question.
\begin{question}
For which integers $k\ge 3$ does there exist a cyclically $k$-connected snark $G$ such that $\omega(G)\neq\omega_{\textrm w}(G)$?
\end{question}
\noindent Note that for $3$-connected snarks the answer coincides with the answer to the following question.
\begin{question}
In a $3$-connected snark, can the expansion of a vertex into a triangle decrease oddness?
\end{question}

\noindent\textbf{Acknowledgements.} We would like to thank Barbora Candráková, Edita Máčajová, Eckhard Steffen, and Martin Škoviera for many fruitful discussions on topics related to even factors in cubic graphs.

The authors acknowledge support from the research grants VEGA 1/0042/14 and VEGA 1/0474/15.

\end{document}